%% file: LSSC-Sozopol-extended.tex
\documentclass{llncs} % 8 pages
\usepackage{graphicx}
\usepackage{caption}

\newcommand{\wt}{\mathrm{wt}}
\newcommand{\dist}{\mathrm{dist}}
\newcommand{\BG}{\mathit{B\!G}}

\usepackage{algorithm}
\usepackage{algpseudocode}

\algrenewcommand\algorithmicrequire{\textbf{Input:}}
\algrenewcommand\algorithmicensure{\textbf{Output:}}

\makeatletter
\let\OldStatex\Statex
\renewcommand{\Statex}[1][3]{%
  \setlength\@tempdima{\algorithmicindent}%
  \OldStatex\hskip\dimexpr#1\@tempdima\relax}
\makeatother
\algrenewcommand{\algorithmiccomment}[1]{$/*$ #1 $*/$}

\usepackage{color}

\begin{document}

\title{Shortest-Path Queries in Planar Graphs \\on GPU-Accelerated Architectures}
\titlerunning{Shortest-Path Queries}  % abbreviated title (for running head)
\author{Guillaume Chapuis \and Hristo Djidjev}
\authorrunning{G. Chapuis and H. Djidjev} % abbreviated author list (for running head)
\institute{Los Alamos National Laboratory, Los Alamos, NM 87545, USA\\
\email{\{gchapuis,djidjev\}@lanl.gov} %,\\ WWW home page:
%\texttt{http://users/\homedir iekeland/web/welcome.html}
}

\maketitle              % typeset the title of the contribution

\begin{abstract}
We develop an efficient parallel algorithm for answering  
shortest-path queries in planar graphs and implement it on a multi-node CPU/GPU clusters. The algorithm uses a divide-and-conquer approach for decomposing the input graph into small and roughly equal subgraphs and constructs a distributed data structure containing shortest distances within each of those subgraphs and between their boundary vertices. For a planar graph with $n$ vertices, that data structure needs $O(n)$ storage per processor and allows queries to be  answered in $O(n^{1/4})$ time.
\end{abstract}
\section{Introduction}
Finding shortest paths (SPs) in graphs has applications in transportation, social network analysis, network routing, and robotics, among others. The problem asks for a path of shortest length between one or more pairs of vertices. There are many algorithm for solving SP problems sequentially. Dijkstra's algorithm~\cite{dijkstra1959note} finds the distances between a source vertex $v$ and all other vertices of the graph in $O(m\log n)$ time, where $n$ and $m$ are the numbers of the vertices and edges of the graph, respectively. It can also be used to find efficiently the distance between a pair of vertices. This algorithm is nearly optimal (within a logarithmic factor), but has irregular structure, which makes it hard to implement efficiently in parallel. Floyd-Warshall's algorithm, on the other hand, finds the distances between all pairs of vertices of the graph in $O(n^3)$ time, which is efficient for dense ($m=\Theta(n^2)$) graphs, has a regular structure good for parallel implementation, but is inefficient for sparse ($m=O(n)$) graphs such as planar graphs. 

In this paper we are considering the query version of the problem. It asks to construct a data structure that will allow to answer any subsequent distance query fast. A distance query asks, given an arbitrary pair of vertices $v,w$, to compute $\dist(v,w)$. This problem has applications in web mapping services such as MapQuest and Google Maps. There is a tradeoff between the size of the data structure and the time for answering a query. For instance, Dijkstra's algorithm gives a trivial solution of the query version of the SP problem with (small) $O(n+m)$ space (for storing the input graph), but large $O(m\log n)$ query time (for running Dijkstra's algorithm with a source the first query vertex). On the other end of the spectrum, Floyd-Warshall's algorithm can be used to construct a (large) $O(n^2)$ data structure (the distance matrix) allowing (short) $O(1)$ query time (retrieving the distance from the data base). However, for very large graphs, the $O(n^2)$ space requirement is impractical. We are interested in an algorithm that needs significantly less than  than $O(n^2)$ space, but will answer queries faster than Disjkstra's algorithm. Our algorithm will use the structure of planar graphs for increased efficiency, as most road networks are planar or near-planar, and will also be highly parallelizable, making use of the features available in modern high-performance clusters and specialized processors such as the %graphics processing unit (GPU).
GPUs.

The query version for shortest path queries in planar graphs was proposed %by Djidjev 
in~\cite{WG-queries} and after that different aspects of the problem were studied by multiple authors, e.g., \cite{Hutchinson:1999,Chen:2000,Kowalik03shortestpath,Mozes:2012}. 
Here we present the first distributed implementation for solving the problem 
%and the first one 
that is designed to make use of the potential for parallelism offered by GPUs. Our solution makes use of the fast parallel algorithm for computing shortest paths in planar graphs from~\cite{ipdps-apsp}, resulting in asymptotically faster and also shown to be efficient in practice.

%previous work, novelty

\section{Preliminaries}
Given a graph $G$ with a weight $\wt(e)$ on each edge $e$, the length of a path $p$ is the sum of the weights of the edges of the path. 
%In this paper we consider nonnegative weights. 
The \textit{single-pair shortest path problem}  (\textit{SPSP}) is, given a pair $v,w$ of vertices of $G$, to find a path between $v$ and $w$, called \textit{shortest path} (SP), with minimum length. The length of that path is called \textit{distance} between $v$ and $w$ and is denoted as $\dist(v,w)$. For any subgraph $H$ of $G$, the distance between $v$ and $w$ in $H$ is denoted as $\dist_H(v,w)$. The \textit{single-source  shortest path problem}  (\textit{SSSP}) is to find SPs from a fixed vertex $v$ to all other vertices of $G$. Finally, the \textit{all-pairs shortest path problem} (\textit{APSP}) is to find SPs between all pairs of vertices. There are distance versions of SPSP, SSSP, and APSP, which are more commonly studied, where the objective is to compute the corresponding distances instead of SPs. Most distance algorithms allow the corresponding SPs to be retrieved in additional time proportional to the number of the edges of the path. In this paper, by SPSP, SSSP, and APSP we mean the distance versions of these problems.

A $k$-partition $\cal P$ of $G$ is a set $V_1,\dots,V_k$ of subsets of $V(G)$, the set of the vertices of $G$, such that $V_i\cap V_j=\emptyset$ if $i\neq j$ and $\bigcup_{i=1}^k V_i=V(G)$. We call the subgraphs of $G$ induced by $V_i$ \textit{components} of $\cal P$. The \textit{boundary} of the partition consists consists of all vertices of $G$ that have at least one neighbor in a different component. We denote by $\BG(G)$ or simply by $\BG$ the subgraph of $G$ induced by the boundary vertices. For any $C\in {\cal P}$, we denote by $B(C)$ the set of all boundary vertices that are from $C$. For any planar graph of $n$ vertices and bounded ($O(1)$ as a function of $n$) vertex degree one can find in $O(n)$ time a $k$-partition $\cal P$ with $|B(C)|=O(\sqrt{n/k})$ for each component $C\in \cal P$.

\section{Algorithm overview and analysis}

Our algorithm works in two modes: preprocessing mode, during which a data structure is computed that allows efficient SP queries, and the query mode that uses that data structure to compute the distance between a query pair of vertices. We assume that the input is a planar graph $G$ of $n$ vertices and bounded vertex degree and the cluster has $p$ nodes.

\subsection{Preprocessing mode}
The preprocessing algorithm (Algorithm~\ref{alg:preprocessing}) has three phases. During the first phase (line 1), the graph is partitioned and each component is assigned to a distinct cluster node. During the second phase (lines 2-5), the APSP problem is solved for each component $C$ independently and in parallel and the computed distance matrix APSP$(C)$ is stored at the same node. Finally, in the third phase (lines 6-10), the boundary graph $\BG$ is constructed and the APSP is solved for $\BG$. That computation is done distributedly such that the distances from vertex $v\in \BG$ to all other vertices of $\BG$ are computed at the node containing $v$, by using Dijkstra's algorithm~\cite{dijkstra1959note}. The computed distance matrix is stored at the node that has done the computations. Hence, at the end of the algorithm, the node $N(C)$ contains two matrices: one containing the SP distances in $C$ and the other containing all SP distances in $\BG$ with source a vertex in $\BG\cap C$.

One can think of $\BG$ as a compressed version of $G$ where the non-boundary vertices are removed, but are implicitly represented in $\BG$ by the information encoded in its edge weights. Note however that the distances APSP$(C)$ (and the corresponding edge weights of $\BG$) are not distances in $G$; the reason is that a shortest path between two vertices $v$ and $w$ from $C$ might pass through vertices not in $C$. Hence the following fact is non-trivial.
\begin{lemma} \label{lem:BG}\cite{ipdps-apsp}
For any two vertices $v,w\in \BG$ the distance between $v$ and $w$ in $\BG$ is equal to the distance between $v$ and $w$ in $G$.
\end{lemma}

We will next estimate the time and space (memory) required to run the algorithm. As $G$ is planar and of bounded vertex degree (as a function of $n$), it can be divided in $O(n) $ time into $k$ parts so that each part has no more than $( n/k)$ vertices and $O(\sqrt{n/k})$ boundary vertices  \cite{frederickson91}. We will estimate the requirements of each phase. Since the maximum amount of coarse-grained parallelism of Algorithm~\ref{alg:preprocessing} is $\min\{p,k\}$, we assume without loss of generalization that $p\leq k$.

Phase 1 requires $O(n)$ running time and $O(n)$ space \cite{frederickson91}.

The complexity of Phase 2 is dominated by the time for computing distances in line 3. We assume that we are using the algorithm from~\cite{ipdps-apsp} that can be implemented efficiently on a GPU-accelerated architecture and has complexity $O(N^{9/4})$. Then Phase~2 requires 
%$O(( k/p ) (n/k)^3)=O(n^3/(pk^2))$ 
$O(( k/p ) (n/k)^{9/4})=O(n^{9/4}/(pk^{5/4}))$ time and 
%$kO((n/k)^2)=O(n^2/k)$ 
$kO((n/k)\sqrt{n/k})=O(n^{3/2}/k^{1/2})$ total space. The space per processor is $kO((n/k)^2)=O(n^2/k)$.

For Phase 3, the number of the vertices of $\BG$ is $k\,O(\sqrt{n/k})=O(\sqrt{nk})$ and the number of the edges is $k\,O((\sqrt{n/k})^2)=O(n)$. One execution of line~8 (for one component $C$) takes $( k/p )|\BG\cap C||E(\BG)|\log(|\BG|)=( k/p ) O(\sqrt{n/k})O(n\log n)$ time and $O(n)$ space. The space needed for one iteration of Step~9 is $|\BG\cap C||\BG|=O(\sqrt{n/k}\sqrt{nk})=O(n)$. Hence Phase~3 requires 
$O(( k/p ) n^{3/2}/k^{1/2}\log n)$ $=O(n^{3/2}k^{1/2}\log n/p)$ time and $O(nk/p)$ space per processor.

Summing up the requirements for Phases 1, 2, and 3, we get $O(n^{9/4}/(pk^{5/4})+n^{3/2}k^{1/2}\log n/p))$ time and $O(n+n^2/(pk)+nk/p)$ space per processor needed for Algorithm~\ref{alg:preprocessing}. Assuming space is more important in this case than time (since nodes have limited memory), we find that $k=n^{1/2}$ minimizes the function $n^2/k+nk$. Hence we have the following result.
\begin{lemma}\label{lem:preprocess}
With $k=\lceil n^{1/2}\rceil$ and $p\leq k$, Algorithm~\ref{alg:preprocessing} runs in $O(n^{7/4}\log n/p)$ time and uses $O(n^{3/2}/p)$ space per processor. With $p=k$, the time and space are $O(n^{5/4})$ and $O(n)$, respectively.
\end{lemma}

The time bound of Lemma \ref{lem:preprocess} is conservative as it doesn't take into account our use of fine-grain parallelism due to multi-threading, e.g., by the GPUs.

\begin{algorithm}
\caption{Preprocessing algorithm \label{alg:preprocessing}}
\begin{algorithmic}[1]
\Require A planar graph $G$
\Ensure A data structure for efficient shortest path queries in $G$
	\Statex \Comment {Partitioning} 
	\State Construct a $k$-partition ${\cal P}$ of $G$ and assign each component $C$ to a distinct node $N(C)$
	\Statex \Comment {Solve the APSP problem for each component}
	\ForAll {components $C\in{\cal P}$} in parallel
		\State Solve APSP for $C$ and save the distances in a table APSP$(C)$
		\State For each pair of boundary vertices $v,w\in C$ define edge $(v,w)$, if not already in 
		\Statex[1] ~~~$G$, and assign a weight $\wt(v,w)=\dist_C(v,w)$
	\EndFor
	\Statex \Comment {Solve the APSP problem for the boundary graph}
	\State Define a boundary graph $\BG$ with vertices all boundary vertices of $G$ and edges as defined in the previous step and store it at each node
	\ForAll {components $C\in{\cal P}$} in parallel
		\State Solve SSSP in $\BG$ for each vertex of $C\cap \BG$ 
		\State Store the distances from all vertices of $C\cap \BG$ to all vertices of $\BG$  in a
		\Statex[1] ~~~table $\mathrm{APSP}_{\BG}(C)$
	\EndFor
\end{algorithmic}
\end{algorithm}

\subsection{Query mode}
The query algorithm (Algorithm~\ref{alg:query}) is based on the fact that if $C_1\neq C_2$, then any path between $v_1$ and $v_2$ should cross both $B(C_1)$ and $B(C_2)$. Let $\pi$ be a shortest path between $v_1$ and $v_2$. Then $\pi$ can be divided into three parts: from $v_1$ to a vertex $b_1$ from $B(C_1)$, from $b_1$ to a vertex $b_2$ on $p$ from $B(C_2)$, and from $b_2$ to $v_2$. Vertices $b_1$ and $b_2$ minimizing the length of $p$ are found as follows: in the loop on lines 2-7, for each $b_2$ an optimal $b_1$ and $\dist(v_1,b_2)$ are found; in lines 10-12 an optimal $b_2$ is found.

\begin{algorithm}
\caption{Query algorithm\label{alg:query}}
\begin{algorithmic}[1]
\Require Vertices $v_1,v_2$ of $G$, a $k$-partition ${\cal P}$ of $G$, tables APSP$(C)$ and APSP$_{\BG}(C)$ for all $C\in {\cal P}$
\Ensure $\dist(v_1,v_2)$
	\State Determine components $C_1$ and $C_2$ such that $v_1\in C_1$, $v_2\in C_2$
	\ForAll {vertices $b_2\in B(C_2)$ }  in parallel 
		\Statex \Comment {Compute $\dist(v_1,b_2)$}
		\State $\dist(v_1,b_2)=\infty$
		\ForAll {vertices $b_1\in B(C_1)$ } 
			\State $\dist(v_1,b_2)=\min\{\dist(v_1,b_2),\dist_{C_1}(v_1,b_1)+\dist_{\BG}(b_1,b_2)\}$ 
		\EndFor
	\EndFor
	\State \textbf{If} $N(C_1)\neq N(C_2)$ \textbf{then} transfer the column of $\mathrm{SP}(C_2)$ corresponding to $v_2$ from $N(C_2)$ to $N(C_1)$.
	\Statex \Comment {Now we can compute $\dist(v_1,v_2)$}
	\State $\dist(v_1,v_2)=\infty$
	\ForAll {vertices $b_2\in B(C_2)$ } 
		\State $\dist(v_1,v_2)=\min\{\dist(v_1,v_2),\dist(v_1,b_2)+\dist_{C_2}(b_2,v_2)\}$
	\EndFor
	\State \textbf{If} $C_1=C_2$ \textbf{then} $\dist(v_1,v_2)=\min\{\dist(v_1,v_2),\dist_{C_1}(v_1,v_2)\}$,
	where the distance $\dist_{C_1}(v_1,v_2)$ is taken from $\mathrm{APSP}(C_1)$.
\end{algorithmic}
\end{algorithm} 

\begin{lemma}
	Algorithm~\ref{alg:query} correctly computes $\dist(v_1,v_2)$ and its running time is $O(n^{1/4})$ with $k=\lceil n^{1/2}\rceil$ and $p\geq \lceil n^{1/4}\rceil$.
\end{lemma}
\begin{proof}
	Let $\pi$ be a shortest path between $v_1$ and $v_2$, let $C_1\neq C_2$, and let $b_1$ be the first vertex along $\pi$ that is on $B(C_1)$ and $\pi_1$ be the subpath of $\pi$ from $v_1$ to $b_1$, let $\pi_2$ be the last vertex along $\pi$ that is on $B(C_2)$ and $\pi_2$ be the subpath of $\pi$ from $b_1$ to $b_2$, and let $\pi_3$ be the subpath of $\pi$ from $b_2$ to $v_2$. Then $\pi_1$ is entirely in $C_1$ and hence $\dist_{C_1}(v_1,b_1)=\dist_{G}(v_1,b_1)$ (note, however, that the distances in APSP$(C_1)$ from $v_1$ to other vertices from $B(C_1)$ may not be correct). Similarly, $\dist_{C_2}(b_2,v_2)=\dist_{G}(b_2,v_2)$. Finally, $\dist_{\BG}(b_1,b_2)=\dist_{G}(b_1,b_2)$ by Lemma~\ref{lem:BG}. Hence lines 5 and 11 use correct values for computing the distances between $v_1$ and $b_2$ and between $b_2$ and $v_2$.
	
	If $C_1=C_2$ (line 13), then a shortest path between $v_1$ and $v_2$ may or may not leave $C_1$. In the first case lines 1-12 compute correctly $\dist(v_1,v_2)$, in the second case $\mathrm{APSP}(C_1)$ contains the correct distance. 
	
	The loop on lines 5-10 takes time $|B(C_1)||B(C_2)|/p=O(\sqrt{n/k}\sqrt{n/k}/p)=O(n/(pk))$, for  $p\leq \min\{k,(n/k)^{1/2}\}$. If  $k=n^{1/2}$ and $p=n^{1/4}$ (the maximum value for which the formula applies), that time becomes $O(n^{1/4})$. The loop in lines 10-12 takes time $O((n/k)^{1/2})=O(n^{1/4})$ for $k=n^{1/2}$.
\end{proof}

Note that using the methodology of \cite{WG-queries}, a more complex implementation of Algorithm~\ref{alg:query} can reduce the query time to logarithmic.
Note also that computation in lines 2-7 can be overlapped with transferring of data in line 8 thereby saving time (upto a factor of two).

\input{implementation-short}
\input{experiments-short}

\section{Conclusion}
We developed and implemented a distributed algorithm for shortest path queries in planar graphs with good scalability. It allows answering SP queries in $O(n^{1/4})$ time by using $O(\sqrt{{n}})$ processors with $O(n)$ space per processor and $O(n^{5/4})$ preprocessing time. Our implementation on 300 node CPU-GPU cluster has preprocessing time of less than 10 seconds using 32 or more nodes and 0.025 milliseconds per query using two nodes. Interesting tasks for future research is implementing a version allowing parallel queries and reducing the query time of the implementation to $O(\log n)$ by using properties of graph planarity. 

%\bibliographystyle{splncs03}
%\bibliography{apsp}

\input{LSSC-Sozopol-extended.bbl}
\end{document}

%% file: implementation-short.tex
\section{Implementation details}
In this section, we describe how the preprocessing and query modes are implemented on a hybrid CPU-GPU cluster. We use a distance matrix to represented both the input graph $G$ and the output. Such a 2-dimensional matrix contains in cell $(i,j)$ the value of the distance from vertex $i$ to vertex $j$. Initially, cell $(i,j)$ contains $\wt(i,j)$ if an edge $(i,j)$ is present in $G$, or infinity otherwise. These values are updated as the algorithm progresses. At the end of the algorithm, cell $(i,j)$ contains $\dist(i,j)$.

In phase 1 of the preprocessing mode, we construct a  $k$-partition of $G$ using the METIS library\cite{karypis1998multilevel}. Based on that partition, we reorder the vertices of $G$ so that vertices from the same component have consecutive indices and boundary vertices of each components have the lowest indices -- see Figure \ref{fig:Adjacency-matrix-after} .

\begin{figure}
\centering
\begin{minipage}{.48\textwidth}
  \centering
  \includegraphics[scale=0.3]{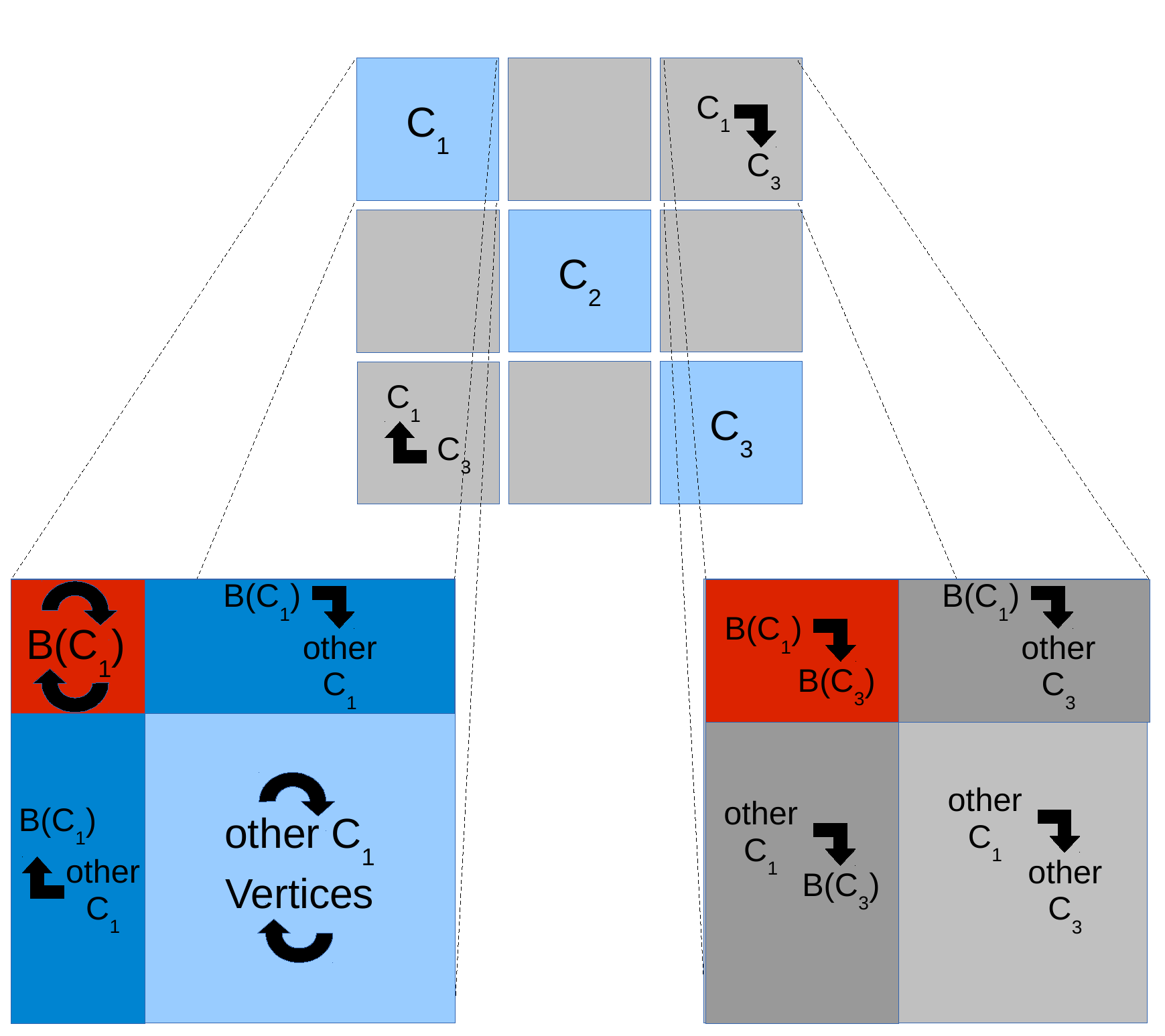}
  \captionof{figure}{Distance matrix after reordering of the vertices. Vertices from the
    same component are stored contiguously starting with boundary vertices. Red submatrices are also part of the boundary distance matrix. Grey submatrices do not generate any computations in preprocessing mode.\label{fig:Adjacency-matrix-after}}
\end{minipage}%
~~
\begin{minipage}{.48\textwidth}
  \centering
  \includegraphics[scale=0.3]{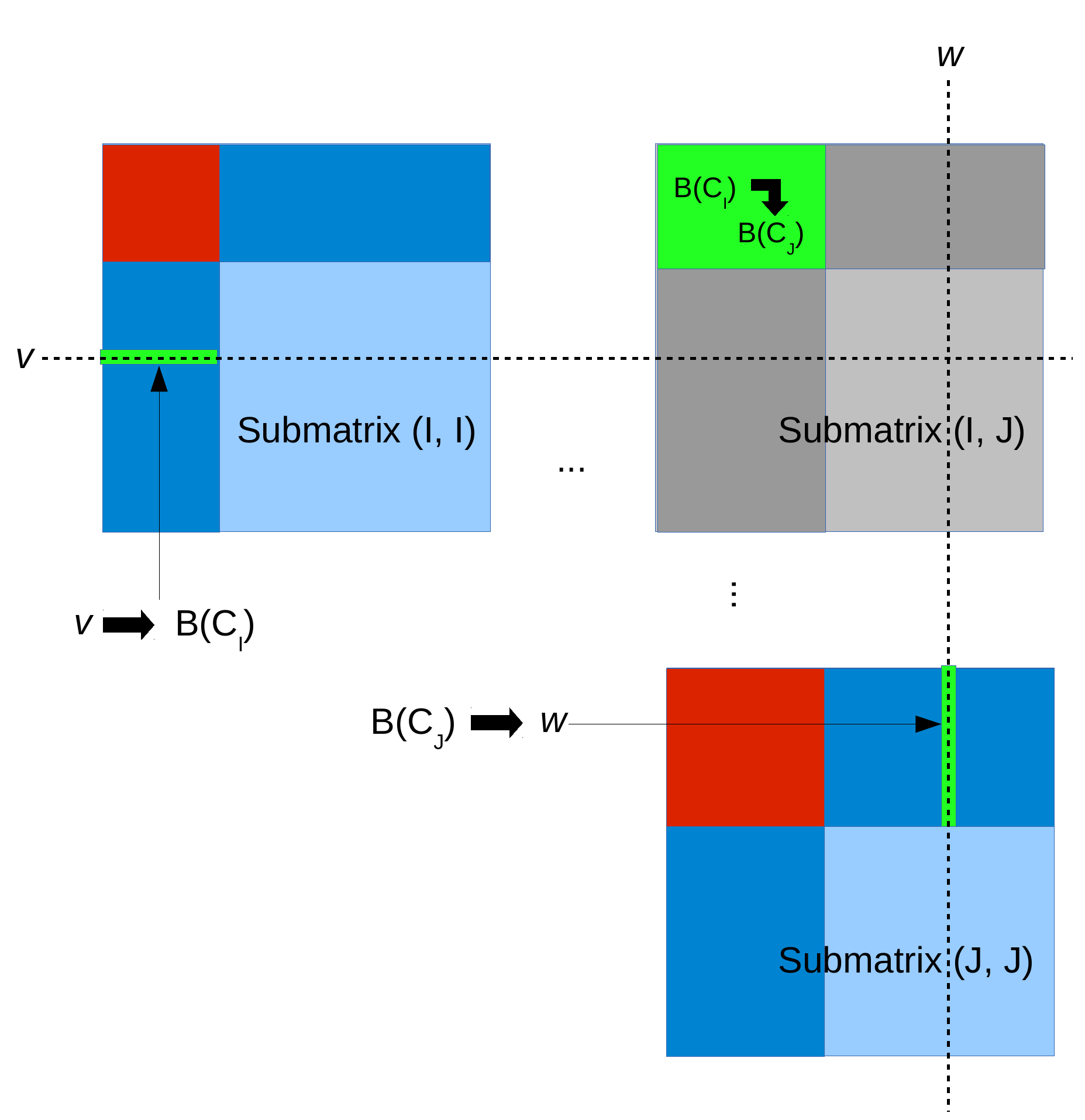}
  \captionof{figure}{The distances required to compute $\dist(v,w)$, shown in green, are scattered in three submatrices: two diagonal ones, for component $I$ and for component $J$, and a non-diagonal submatrix $(I,J)$.\label{fig:values_required}}
\end{minipage}
\end{figure}

In phase 2, we compute the shortest distances within each of the components. For $k$ components, this phase gives a total $k$ independent tasks that can be executed in parallel. 
Computations at this phase are already balanced across nodes as components contain roughly the same number of vertices and the APSP algorithm from \cite{ipdps-apsp} ensures the same $O(N^{9/4})$ complexity with respect to the number of nodes.

Finally, phase 3 consists in computing the shortest distances within the boundary graph using Dijkstra's algorithm. Computations at this phase may be imbalanced between nodes for two reasons. First, the number of boundary vertices in two components may differ and, second, the complexity of Dijkstra's algorithm does not solely depend on the number of vertices in the graph, but also on the number of edges, which may vary even more than the number of vertices between two components' boundary graphs.

In the query mode, we are interested in finding $\dist(v,w)$, where $v$ and $w$ are from components $I$ and $J$, respectively. The required values for that computation are scattered in three submatrices, as illustarted in Figure~\ref{fig:values_required}.
For such a query, assuming $k=p$, node $i$, holding the required values from diagonal submatrix $I$ and non-diagonal submatrix $(I, J)$, will be in charge of the computations. Required values from diagonal submatrix $J$ are held by node $j$ and need to be transfered to node $i$.

%% file: experiments-short.tex
\section{Experimental evaluation}
In this section we describe experiments designed to test our algorithm and its implementation. Specifically, we are going to test the strong scaling properties by running our code  on a fixed graph size and a varying number $p$ of cluster nodes and number $k$ of components. 
All computations are run on a 300 node cluster. Each cluster node is comprised of 2 x Eight-Core Intel Xeon model E5-2670 @ 2.6 GHz and two GPGPU Nvidia Tesla M2090 cards connected to PCIe-2.0 x16 slots. In order to make full use of the available GPUs, each node is assigned at least two graph components so that the two associated diagonal submatrices can be computed simultaneously on the two GPUs.

%\subsection{Strong scaling study}   
For the strong-scaling experiment, the graph size is fixed to 256k vertices. Preprocessing and queries are run with increasing numbers of nodes ranging from $4$ to $64$. Each node handles $2$ components (one per available GPU); therefore the number of components $k$ ranges from $8$ to $128$. 

\begin{figure}
\centering
\begin{minipage}{.5\textwidth}
  \centering
  \includegraphics[bb=40bp 250bp 560bp 540bp,clip,scale=.35]{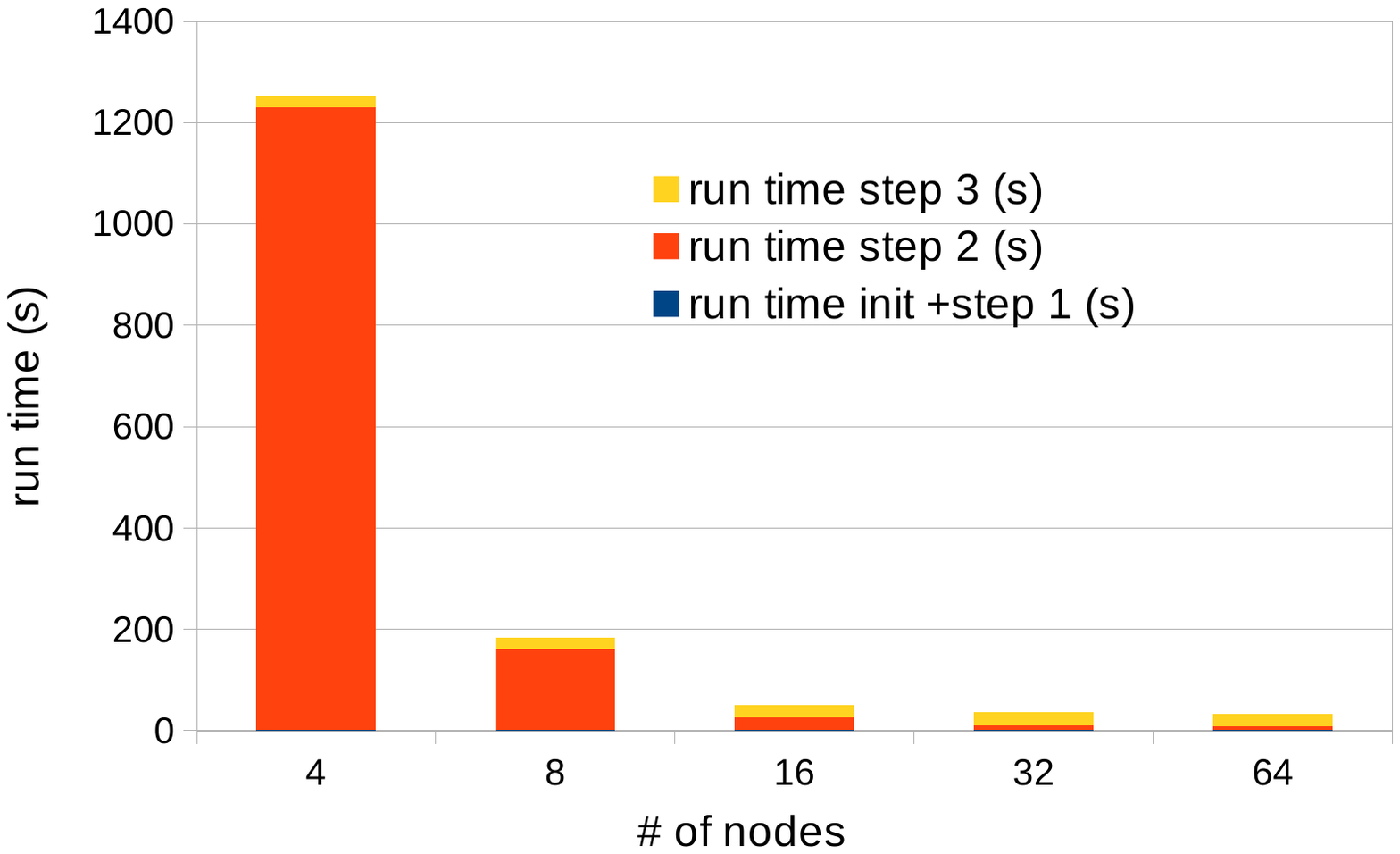}
  \captionof{figure}{Preprocessing run times for a fixed graph size of 256k vertices and increasing number of nodes.}
  \label{fig:exp1_preprocessing}
\end{minipage}%
~~
\begin{minipage}{.5\textwidth}
  \centering
  \includegraphics[bb=40bp 250bp 600bp 540bp,clip,scale=.35]{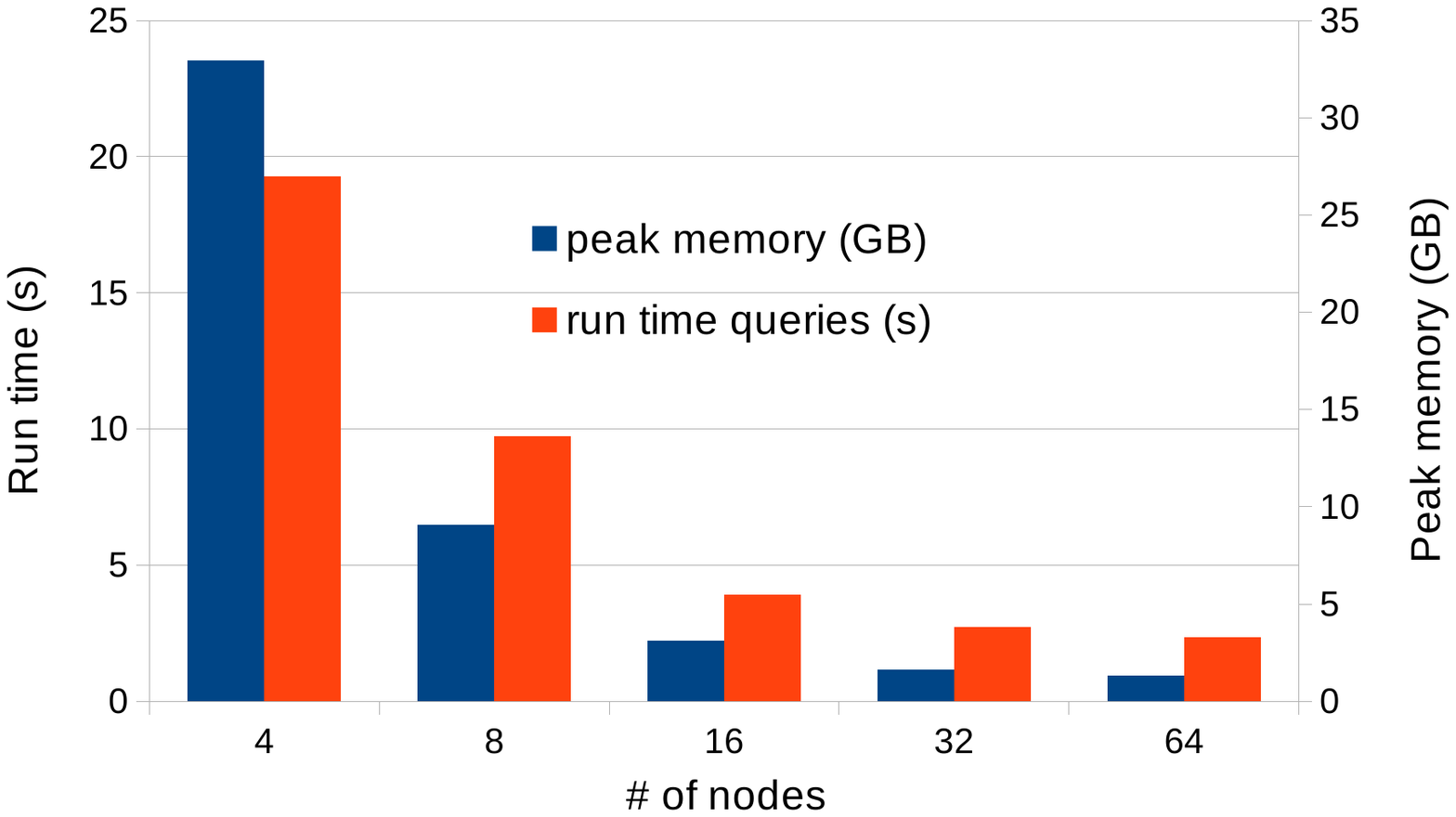}
  \captionof{figure}{Peak memories and run times for 10k queries for a fixed graph size of 256k vertices and increasing number of parts/processors.}
  \label{fig:exp1_queries}
\end{minipage}
\end{figure}

Figure~\ref{fig:exp1_preprocessing} shows the run times for the preprocessing mode. For low numbers of nodes and thus low values of $k$, preprocessing time is dominated by step 2 - the computation of the shortest distances within each component - since lower $k$ values means larger components. For higher numbers of nodes and thus higher values of $k$, preprocessing time becomes dominated by step 3 - the computation of the boundary graph - as more components mean higher numbers of incident edges and thus larger boundary graphs. 
%With $16$ nodes and $k=32$, computations are balanced between step 2 and step 3; this value is therefore optimal for run time performances.  
Note that while the figure seems to show supralinear speedup, that is not the case (and similarly for the memory usage).
The reason is that, with increasing the number of processors $p$, the number $k$ of parts is increased too (as it is tied to $p$ in this implementation) and hence the complexity of the algorithm is also reduced.

Figure~\ref{fig:exp1_queries} shows the query times and peak memory usage per node. The run times are given for $10,000$ queries from random sources to random targets. Note that in the query mode only fine-grain (node-level) parallelism is used, while multiple nodes are still needed for distributed storage and, optionally, to handle multiple queries in parallel (not implemented in the current version).
For the memory usage, the optimal value for $k$, theoretically expected to be $\sqrt{n}$ -- or $512$ for this instance -- is not reached in this experiment since $k$ only goes up to $128$. We can however see that peak memory usage per node is still dropping with increasing values of $k$ up to 128. The query times in the figure vary from about 2 milliseconds per query for $k=8$ to 0.25 milliseconds for $k=128$. Compared with the Boost library implementation of Dijkstra's algorithm, our implementation answers queries on the largest instances about 1000 times faster.  

%\begin{figure}
%\centering
%\includegraphics[bb=40bp 250bp 560bp 540bp,clip,scale=.35]{./exp1_memory}
%\captionof{figure}{Peak memory usage per node for a fixed graph size of 256k vertices and increasing number of nodes.}
%\label{fig:exp1_memory}
%\end{figure}
%